\documentclass{article}
\usepackage{graphicx} 
\usepackage{amssymb}
\usepackage{amsmath,amsthm, bbm, mathdots}
\usepackage{color, xcolor}
\usepackage{graphicx}
\usepackage{setspace}
\usepackage{xspace}
\usepackage{multirow}
\usepackage{array}
\usepackage{complexity}
\usepackage{geometry}
\usepackage{mathtools}
\usepackage{algorithm,algpseudocode}
\usepackage{dirtytalk}
\usepackage{tikz}
\usepackage
{hyperref}
\usepackage[capitalize]{cleveref}
\usepackage{dsfont}

\hypersetup{colorlinks=true,urlcolor=blue,linkcolor=magenta,citecolor=[rgb]{.42,.56,.14},}

\usepackage{indentfirst}
\makeatletter
\renewcommand{\@seccntformat}[1]{\csname the#1\endcsname.\quad}
\makeatother

\theoremstyle{plain}
\newtheorem{theorem}{Theorem}
\newtheorem{corollary}[theorem]{Corollary}
\newtheorem{proposition}[theorem]{Proposition}
\newtheorem{lemma}[theorem]{Lemma}

\theoremstyle{remark}

\theoremstyle{plain}
\newclass{\DNF}{DNF}
\newclass{\DNFs}{DNFs}
\newclass{\ACzero}{AC^0}
\newclass{\TCzero}{TC^0}

\renewcommand{\E}{\mathop{\mathbb{E}}}

\newcommand{\eps}{\varepsilon}

\newcommand{\pmone}{\{\pm 1\}}

\newcommand{\calX}{\mathcal{X}}
\newcommand{\calY}{\mathcal{Y}}

\renewcommand{\epsilon}{\varepsilon}

\newcommand{\rk}{\mathsf{rk}}

\usepackage{amssymb}

\title{An XOR Lemma for Deterministic Communication Complexity}

\author{Siddharth Iyer\thanks{Supported by NSF award 2131899.} \\ siyer@cs.washington.edu \and Anup Rao\footnotemark[1] \\ anuprao@cs.washington.edu}

\begin{document}

\maketitle
\begin{abstract}
    We prove a lower bound on the communication complexity of computing the $n$-fold xor of an arbitrary function $f$, in terms of the communication complexity and rank of $f$. We prove that  $D(f^{\oplus n}) \geq n \cdot \Big(\frac{\Omega(D(f))}{\log \rk(f)}  -\log \rk(f)\Big )$, where here $D(f), D(f^{\oplus n})$ represent the deterministic communication complexity, and $\rk(f)$ is the rank of $f$. Our methods involve a new way to use information theory to reason about deterministic communication complexity. 
\end{abstract}

\section{Introduction}
How is the complexity of computing a Boolean function $f$ on 1 input related to the complexity of computing  $f$ on $n$ inputs? In this work, we give new lower bounds for the deterministic communication complexity of computing $f$ on $n$ inputs, making the first progress on this question in many years. We refer the reader to the textbooks \cite{Kushilevitz-Nisan, RaoYehudayoff-text} for the broader context surrounding these problems and the model of communication complexity.

Given a function $f: \mathcal{X} \times \mathcal{Y} \rightarrow \{0,1\}$, define the functions
\begin{align*}
     f^n(x_1,\dotsc,x_n,y_1,\dotsc, y_n) &= f(x_1,y_1),f(x_2,y_2),\dotsc,f(x_n,y_n),\\
     f^{\oplus n}(x_1,\dotsc,x_n,y_1,\dotsc, y_n) &= f(x_1,y_1)\oplus f(x_2,y_2)\oplus \dotsc \oplus f(x_n,y_n).
     \end{align*}   
So, $f^n$ computes $f$ on $n$ distinct inputs, and $f^{\oplus n}$ computes the parity of the outputs of $f$. Because every protocol computing $f^n$ is also a protocol for computing $f^{\oplus n}$, the complexity of computing $f^{\oplus n}$ can only be smaller. An important example to keep in mind is when  $x,y$ are bits and $f(x,y)=x \oplus y$. Then the communication complexity of $f$ and $f^{\oplus n}$ are both $2$, so there is no increase in the complexity of the xor for such functions.

The communication complexity of a function $f$ is related to the number of \emph{monochromatic rectangles} needed to cover the inputs to $f$. A monochromatic rectangle is a pair $A \subseteq \mathcal{X}, B \subseteq \mathcal{Y}$ such that $f$ is constant when restricted to $A \times B$. Let $D(f)$ denote the communication complexity of $f$, and let $C(f)$ denote the minimum number of monochromatic rectangles needed to cover the inputs of $f$. It is a standard fact that $D(f) \geq \log C(f)$. Prior to our work, the best known result concerning the complexity of computing these functions was proved by Feder, Kushilevitz, Naor and Nisan \cite{FKNN}, who showed that when $\sqrt{D(f)} > \log \log (|\mathcal{X}|\cdot  |\mathcal{Y}|)$, $D(f^n)$ grows with $n$:
\begin{theorem}[\cite{FKNN}] \label{FKNN}
$D(f^n) \geq \log C(f^{n}) \geq n \cdot (\sqrt{D(f)} - \log \log (|\mathcal{X}|\cdot  |\mathcal{Y}|))$.
\end{theorem}

Another important parameter of $f$ is its rank. The function $f$ can be viewed as a Boolean matrix $M$ whose $(x,y)$'th entry is $(-1)^{f(x,y)}$. We write $\rk(f)$ to denote the rank of this matrix. Because $M$ has $\pm 1$ entries, it can have at most $2^{\rk(f)}$ distinct rows and at most $2^{\rk(f)}$ distinct columns. This observation leads to the following corollary of \Cref{FKNN}:

\begin{corollary}
$D(f^n) \geq n \cdot (\sqrt{D(f)} - \log \rk(f) -1)$.
\end{corollary}

There have been a number of results concerning the randomized communication complexity of $f^n$ and $f^{\oplus n}$ in recent years. These results rely on definitions from information complexity and simulations of protocols that have small information complexity.
See 
\cite{KalyanasundaramSchnitger87, Razborov92, Raz, CSWY, BJKS, JRS, BBCR, HJMR, BR, Braverman, Kol, Sherstov, JPY, BRWY, BRWY-bounded-round, huacheng, GKR, RR, IR}.
However, communication complexity is a model where the deterministic and randomized complexity can be quite far from each other. For example, the randomized communication complexity of the equality function is a constant, but there is no deterministic protocol that beats the performance of the trivial protocol. 

In fact, a number of connections between the model of communication complexity and other models of computation are only meaningful when using deterministic protocols. A good example is the connection between circuit depth and communication complexity observed by Karchmer and Wigderson \cite{KW}. The randomized communication complexity of every Karchmer-Wigderson game is small, because the game can efficiently be solved by hashing. So, lower bounds on circuit depth can only be obtained by studying deterministic communication complexity. Karchmer, Raz and Wigderson \cite{KRW} conjectured that the communication complexity of this problem increases when the function is composed with itself. Recently, there have been attempts toward this conjecture and on understanding Karchmer-Wigderson games using ideas from information theory \cite{GMWW, MW}. If the conjecture is true, this would imply that there is no way to simulate every polynomial time algorithm in logarithmic time with a parallel algorithm. Achieving such tantalizing results motivates us to study the questions about deterministic communication complexity we consider in this paper.

Before the present paper, techniques from information theory had not led to results about the deterministic communication complexity of $f^n$ or $f^{\oplus n}$. That is because known methods to simulate protocols with small information lead to simulations that introduce errors, even if the protocols being simulated do not make errors. In the present paper, we use information theory to obtain results about deterministic communication complexity without introducing any errors. That is the key technical contribution of our work. Our proofs are short, but they circumvent a barrier to applying information theory in the setting of deterministic communication protocols.

 Lov{\'a}sz and Saks \cite{lovaszsaks} conjectured that there is a constant $c$ such that $D(f) \leq (\log \rk(f))^c)$. This is called the \emph{log-rank} conjecture. To date, the best known upper bound is $D(f) \leq \sqrt{\rk(f)}$ \cite{Lovett, sudakov2023matrix}, and it is known that there are $f$ with $D(f) \geq (\log \rk(f))^{2-o(1)}$ \cite{GPW}.  Recall that $D(f) \geq \log \rk(f)$. Our main result gives stronger lower bounds when $D(f) \gg (\log \rk(f))^2$:
\begin{theorem}\label{main}
    $D(f^{\oplus n}) \geq \log C(f^{\oplus n})\geq n \cdot \Big(\frac{\Omega(D(f))}{\log \rk(f)}  -\log \rk(f)\Big )$.
\end{theorem}
In comparison to \Cref{FKNN}, our result gives lower bounds even for computing the xor $f^{\oplus n}$. The key new step of our proof is the following theorem, whose proof uses the sub-additivity of entropy in an essential way:
\begin{theorem}\label{thm:subadditivity}
    If $f^{\oplus n}$ has a monochromatic rectangle of size $2^k$, then $f$ has a monochromatic rectangle of size $2^{k/n-2}$.
\end{theorem}

The above theorem allows us to use a monochromatic rectangle of large density in $f^{\oplus n}$ to find a monochromatic rectangle of even larger density in $f$. Combined with some reasoning about the rank of the function, we are able to use \Cref{thm:subadditivity} to obtain a deterministic protocol that proves \Cref{main}. In the rest of this paper, we give the details of the proofs of these two theorems.

\section{Preliminaries and Notation}
For a variable $X=X_1,\dotsc, X_n$, we write $X_{<i}$ to denote $X_1,\dotsc, X_{i-1}$. We define $X_{>i}$ similarly. All logarithms are taken base $2$. We recall some basic definitions regarding entropy of random variables. Let $A$ be a random variable distributed according to $p(a)$. The entropy of $A$ is defined as
\[H(A) := \E_{p(a)}\bigg[\log\frac{1}{p(a)}\bigg].\]

\begin{proposition}\label{fact:entropy-upper-bound}
        For any random variable $A$ with finite support, we have $H(A) \leq \log |\mathsf{supp}(A)|$, with equality if $A$ is distributed according to the uniform distribution. 
\end{proposition}
    
If $A$ and $B$ are two jointly distributed random variables distributed according to $p(ab)$ then the entropy of $A$ conditioned on $B$ is defined as
\[H(A\vert B) := \E_{p(a,b)}\bigg[\log\frac{1}{p(a\vert b)}\bigg].\]

The entropy of jointly distributed random variables satisfy the chain rule:
\[H(A,B) = H(A) + H(B\vert A).\]
Additionally, it is known that the conditional entropy of a random variable cannot exceed its entropy.

\begin{lemma}\label{lem:conditioning}
    For any two jointly distributed random variables,  $A,B$, we have $H(A\vert B) \leq H(A)$, with equality if $A$ and $B$ are independent.
\end{lemma}

We need the following basic fact about rank:
\begin{proposition}\label{fact:rank-subadd}
For any two matrices $A_1$ and $A_2$, we have
$\rk(A_1+A_2) \leq \rk(A_1) + \rk(A_2)$.
\end{proposition}

We need the following lemma that shows that a protocol with a small number of leaves can be computed by a protocol with small communication (see \cite{RaoYehudayoff-text}, Theorem 1.7).
\begin{lemma}\label{lem:balancing}
    If $\pi$ is a deterministic protocol with $\ell$ leaves, there exists a deterministic protocol computing $\pi(x,y)$ with communication at most $\lceil2\log_{3/2} \ell\rceil$. 
\end{lemma}

\section{Proof of \Cref{thm:subadditivity}}

Let $R$ be a monochromatic rectangle for $f^{\oplus n}$ of size $2^k$, and let $(X,Y) \in R$ be uniformly random. Because $R$ is a rectangle, $X$ and $Y$ are  independent. Using the chain rule, we get
\begin{align*}
    k= H(XY) &= H(X) + H(Y) \tag{because $X,Y$ are independent}\\
    &= \sum_{i=1}^n  H(X_i\vert X_{< i}) + H(Y_i\vert Y_{> i}) \tag{by the chain rule}\\
    &= \sum_{i=1}^n  H(X_i\vert X_{< i} Y_{> i}) + H(Y_i\vert X_{<i}Y_{> i} X_i) \tag{because $X,Y$ are independent}\\
    &= \sum_{i=1}^n  H(X_iY_i\vert X_{< i} Y_{> i}) \tag{by the chain rule}.
\end{align*}
This implies there exist $i,x_{<i}, y_{>i}$ such that 
\begin{align*}
    H(X_i Y_i \vert x_{<i} y_{>i}) \geq k/n.
\end{align*}
Define the random variables  $U=f(x_1,Y_1)\oplus\ldots \oplus f(x_{i-1},Y_{i-1})$ and $V=f(X_{i+1},y_{i+1})\oplus\ldots \oplus f(X_{n},y_{n})$. By the chain rule, and since $U,V$ are bits, we get
\begin{align*}
 H(X_i Y_i \vert x_{<i} y_{>i} U V)+2 &\geq
    H(X_i Y_i \vert x_{<i} y_{>i} U V)+ H(UV \vert x_{<i} y_{>i})\\ &= H(X_i Y_i UV\vert x_{<i} y_{>i} )\\&\geq H(X_i Y_i \vert x_{<i} y_{>i} ) \\ &\geq  k/n,
\end{align*}
so there is some fixed value of $u,v$ such that
\begin{align*}
    H(X_i Y_i \vert x_{<i} y_{>i} u v) \geq  k/n-2.
\end{align*}

The desired rectangle is the support of $(X_i,Y_i)$ conditioned on this fixed value of $(x_{<i},y_{>i}, u,v)$, which we call $T$. Because $(X,Y)$ is distributed uniformly in $R$, the distribution of $(X_i,Y_i)$ conditioned on $(x_{<i},y_{>i}, u,v)$ is a product distribution, and so $T$ is a rectangle. By \Cref{fact:entropy-upper-bound}, $|T| \geq 2^{k/n-2}$.  Because each input $(x_i,y_i) \in T$ corresponds to some input $(x,y)\in R$ with $f^{\oplus n}(x,y)$ fixed, and we have fixed $x_{< i},y_{> i}$ and the xor of the function value in the first $i-1$ as well as the last $n-i$ coordinates, $f(x_i,y_i)$ is determined within $T$, and $T$ is a monochromatic rectangle of $f$.

\section{Proof of \Cref{main}}
The proof uses \Cref{thm:subadditivity} and standard ideas along the lines of \cite{NW} to obtain a protocol for $f$. 
We shall prove that  $f$ has a protocol tree whose number of leaves is bounded by 
\begin{align} \label{leaves}
2^{O((\log C(f^{\oplus n})^{1/n}+ \log \rk(f))\cdot \log \rk(f))}
\end{align}
By applying 
\Cref{lem:balancing} to this protocol, we obtain a protocol with communication 
\[O\bigg((\log C(f^{\oplus n})^{1/n} +\log \rk(f)) \log \rk(f)\bigg) \geq D(f),\] which proves that
$$\log C(f^{\oplus n})^{1/n} \geq \frac{\Omega(D(f))}{\log \rk(f)} - \log \rk(f),$$ yielding the theorem. 

We prove the bound by induction on $|\mathcal{X}| \cdot |\mathcal{Y}|$ and $\rk(f)$. If $\rk(f) < 5$, or $|\mathcal{X}| \cdot |\mathcal{Y}|\leq 1$, we obtain a protocol with a constant number of leaves. Otherwise, by averaging, $f^{\oplus n}$ has a monochromatic rectangle of size  $$\frac{|\calX|^n\cdot |\calY|^n} { C(f^{\oplus n})}.$$
    \Cref{thm:subadditivity} then implies that $f$ contains a monochromatic rectangle $R$  of size at least $$\frac{ |\calX|\cdot |\calY|}{ 4 \cdot C(f^{\oplus n})^{1/n}} .$$
    We can use $R$ to partition the matrix corresponding to $f$ as follows
    \[\begin{bmatrix}
R & A \\
B & Z
\end{bmatrix}.\]
Since $R$ has rank  $1$, we have
\begin{align*}
    \rk(f) &\geq \rk\left(\begin{bmatrix}
0 & A \\
B & Z
\end{bmatrix}\right) - 1 \tag{\Cref{fact:rank-subadd}}  \\
&\geq \rk\left(\begin{bmatrix}
0 & A
\end{bmatrix}\right) + \rk\left(\begin{bmatrix}
0 \\
B
\end{bmatrix}\right)  - 1 \tag{by Gaussian elimination}\\
&\geq \rk\left(\begin{bmatrix}
R & A
\end{bmatrix}\right) + \rk\left(\begin{bmatrix}
R \\
B
\end{bmatrix}\right)  - 3. \tag{\Cref{fact:rank-subadd}} 
\end{align*}
So, we must have either \begin{align} \label{alicesends}
    \rk(\begin{bmatrix}
R & A
\end{bmatrix}) \leq (\rk(f)+3)/2,
\end{align} or
\begin{align*} 
    \rk \Big(\begin{bmatrix}
R \\ B
\end{bmatrix}\Big ) \leq (\rk(f)+3)/2.
\end{align*}
If  \Cref{alicesends} holds, Alice sends a bit to Bob indicating whether her input is consistent with $R$. Otherwise, Bob sends a bit indicating whether his input is consistent with $R$. Without loss of generality, assume that \Cref{alicesends} holds.

Let $f_0$ and $f_1$ denote the sub-functions of $f$ obtained by restricting to $\begin{bmatrix}
R & A
\end{bmatrix}$ and $\begin{bmatrix}
B & Z
\end{bmatrix}$ respectively. Since every rectangle cover of $f^{\oplus n}$ yields a rectangle cover of $f_0^{\oplus n}$ and a rectangle cover of $f_1^{\oplus n}$, we have $$\max\{C(f_0^{\oplus n}), C(f_1^{\oplus n})\} \leq C(f^{\oplus n}).$$


If Alice's input is consistent with $R$, 
we may repeat the argument with the function $f_0$ which satisfies $\rk(f_0) \leq (\rk(f) + 3)/2 \leq 4\rk(f)/5$, so long as $\rk(f) \geq 5$. Otherwise, if Alice's input is inconsistent with $R$,  
we repeat the argument with the function $f_1$ which has at most $$|\calX|\cdot |\calY|\cdot \Big(1 - \frac{1}{4 \cdot C(f^{\oplus n})^{1/n}}\Big )$$ inputs.

The number of recursive steps where the rank reduces by a factor of $4/5$ is at most $O(\log \rk(f))$. Moreover, since the  matrix corresponding to $f$ has at most $2^{\rk(f)}$ distinct rows and columns, the number of steps where the input space shrinks by a factor of $(1 - \frac{1}{4 \cdot C(f^{\oplus n})^{1/n}})$ is at most $8\cdot \rk(f)\cdot C(f^{\oplus n})^{1/n} $. That is because after so many steps the number of inputs is at most \[2^{2\cdot \rk(f)}\cdot \Big (1 - \frac{1}{4 \cdot C(f^{\oplus n})^{1/n}}\Big )^{8\rk(f)\cdot C(f^{\oplus n})^{1/n}} \leq 2^{2\rk(f)}\cdot e^{-2\rk(f)}< 1.\] 

The number of leaves in the protocol we have designed is at most \begin{align*}
    \binom{8\cdot \rk(f)\cdot C(f^{\oplus n})^{1/n} + O(\log \rk(f))}{O(\log \rk(f))} 
    &\leq 2^{O((\log C(f^{\oplus n})^{1/n}+ \log \rk(f)) \log \rk(f))},\end{align*}
    since $C(f^{\oplus n}) \geq 1$. This proves \Cref{leaves}.

\bibliographystyle{alpha}

\bibliography{ref}

\begin{thebibliography}{BRWY13b}

\bibitem[BBCR10]{BBCR}
Boaz Barak, Mark Braverman, Xi~Chen, and Anup Rao.
\newblock How to compress interactive communication.
\newblock In {\em Proceedings of the Forty-Second ACM Symposium on Theory of Computing}, STOC '10, page 67–76, New York, NY, USA, 2010. Association for Computing Machinery.

\bibitem[BJKS02]{BJKS}
Ziv Bar{-}Yossef, T.~S. Jayram, Ravi Kumar, and D.~Sivakumar.
\newblock An information statistics approach to data stream and communication complexity.
\newblock In {\em 43rd Symposium on Foundations of Computer Science {(FOCS} 2002), 16-19 November 2002, Vancouver, BC, Canada, Proceedings}, pages 209--218. {IEEE} Computer Society, 2002.

\bibitem[BR11]{BR}
Mark Braverman and Anup Rao.
\newblock Information equals amortized communication.
\newblock In {\em 2013 IEEE 54th Annual Symposium on Foundations of Computer Science}, pages 748--757, Los Alamitos, CA, USA, oct 2011. IEEE Computer Society.

\bibitem[Bra15]{Braverman}
Mark Braverman.
\newblock Interactive information complexity.
\newblock {\em SIAM Journal on Computing}, 44(6):1698--1739, 2015.

\bibitem[BRWY13a]{BRWY-bounded-round}
Mark Braverman, Anup Rao, Omri Weinstein, and Amir Yehudayoff.
\newblock Direct product via round-preserving compression.
\newblock In Fedor~V. Fomin, Rusins Freivalds, Marta~Z. Kwiatkowska, and David Peleg, editors, {\em Automata, Languages, and Programming - 40th International Colloquium, {ICALP} 2013, Riga, Latvia, July 8-12, 2013, Proceedings, Part {I}}, volume 7965 of {\em Lecture Notes in Computer Science}, pages 232--243. Springer, 2013.

\bibitem[BRWY13b]{BRWY}
Mark Braverman, Anup Rao, Omri Weinstein, and Amir Yehudayoff.
\newblock Direct products in communication complexity.
\newblock In {\em 54th Annual {IEEE} Symposium on Foundations of Computer Science, {FOCS} 2013, 26-29 October, 2013, Berkeley, CA, {USA}}, pages 746--755. {IEEE} Computer Society, 2013.

\bibitem[CSWY01]{CSWY}
Amit Chakrabarti, Yaoyun Shi, Anthony Wirth, and Andrew~Chi{-}Chih Yao.
\newblock Informational complexity and the direct sum problem for simultaneous message complexity.
\newblock In {\em 42nd Annual Symposium on Foundations of Computer Science, {FOCS} 2001, 14-17 October 2001, Las Vegas, Nevada, {USA}}, pages 270--278. {IEEE} Computer Society, 2001.

\bibitem[FKNN95]{FKNN}
Tom\'{a}s Feder, Eyal Kushilevitz, Moni Naor, and Noam Nisan.
\newblock Amortized communication complexity.
\newblock {\em SIAM Journal on Computing}, 24(4):736--750, 1995.

\bibitem[GKR16]{GKR}
Anat Ganor, Gillat Kol, and Ran Raz.
\newblock Exponential separation of information and communication for boolean functions.
\newblock {\em J. ACM}, 63(5), nov 2016.

\bibitem[GMWW17]{GMWW}
Dmitry Gavinsky, Or~Meir, Omri Weinstein, and Avi Wigderson.
\newblock Toward better formula lower bounds: The composition of a function and a universal relation.
\newblock {\em {SIAM} J. Comput.}, 46(1):114--131, 2017.

\bibitem[GPW18]{GPW}
Mika G\"{o}\"{o}s, Toniann Pitassi, and Thomas Watson.
\newblock Deterministic communication vs. partition number.
\newblock {\em SIAM Journal on Computing}, 47(6):2435--2450, 2018.

\bibitem[HJMR10]{HJMR}
Prahladh Harsha, Rahul Jain, David McAllester, and Jaikumar Radhakrishnan.
\newblock The communication complexity of correlation.
\newblock {\em IEEE Transactions on Information Theory}, 56(1):438--449, 2010.

\bibitem[IR24]{IR}
Siddharth Iyer and Anup Rao.
\newblock Xor lemmas for communication via marginal information.
\newblock In {\em Proceedings of the 56th Annual ACM Symposium on Theory of Computing}, STOC 2024, page 652–658, New York, NY, USA, 2024. Association for Computing Machinery.

\bibitem[JPY12]{JPY}
Rahul Jain, Attila Pereszl{\'{e}}nyi, and Penghui Yao.
\newblock A direct product theorem for the two-party bounded-round public-coin communication complexity.
\newblock In {\em 53rd Annual {IEEE} Symposium on Foundations of Computer Science, {FOCS} 2012, New Brunswick, NJ, USA, October 20-23, 2012}, pages 167--176. {IEEE} Computer Society, 2012.

\bibitem[JRS03]{JRS}
Rahul Jain, Jaikumar Radhakrishnan, and Pranab Sen.
\newblock A direct sum theorem in communication complexity via message compression.
\newblock In Jos C.~M. Baeten, Jan~Karel Lenstra, Joachim Parrow, and Gerhard~J. Woeginger, editors, {\em Automata, Languages and Programming}, pages 300--315, Berlin, Heidelberg, 2003. Springer Berlin Heidelberg.

\bibitem[KN97]{Kushilevitz-Nisan}
Eyal Kushilevitz and Noam Nisan.
\newblock {\em Communication complexity}.
\newblock Cambridge University Press, Cambridge, 1997.

\bibitem[Kol16]{Kol}
Gillat Kol.
\newblock Interactive compression for product distributions.
\newblock In {\em Proceedings of the Forty-Eighth Annual ACM Symposium on Theory of Computing}, STOC '16, page 987–998, New York, NY, USA, 2016. Association for Computing Machinery.

\bibitem[KRW95]{KRW}
Mauricio Karchmer, Ran Raz, and Avi Wigderson.
\newblock Super-logarithmic depth lower bounds via the direct sum in communication complexity.
\newblock {\em Comput. Complex.}, 5(3/4):191--204, 1995.

\bibitem[KW90]{KW}
Mauricio Karchmer and Avi Wigderson.
\newblock Monotone circuits for connectivity require super-logarithmic depth.
\newblock {\em SIAM Journal on Discrete Mathematics}, 3(2):255--265, 1990.

\bibitem[Lov14]{Lovett}
Shachar Lovett.
\newblock Communication is bounded by root of rank.
\newblock In {\em Proceedings of the Forty-Sixth Annual ACM Symposium on Theory of Computing}, STOC '14, page 842–846, New York, NY, USA, 2014. Association for Computing Machinery.

\bibitem[LS88]{lovaszsaks}
L{\'a}szl{\'o}~Mikl{\'o}s Lov{\'a}sz and Michael~E. Saks.
\newblock Lattices, mobius functions and communications complexity.
\newblock {\em [Proceedings 1988] 29th Annual Symposium on Foundations of Computer Science}, pages 81--90, 1988.

\bibitem[MW19]{MW}
Or~Meir and Avi Wigderson.
\newblock Prediction from partial information and hindsight, with application to circuit lower bounds.
\newblock {\em Comput. Complex.}, 28(2):145--183, 2019.

\bibitem[NW95]{NW}
Noam Nisan and Avi Wigderson.
\newblock On rank vs. communication complexity.
\newblock {\em Comb.}, 15(4):557--565, 1995.

\bibitem[Raz92]{Razborov92}
Alexander~A. Razborov.
\newblock On the distributional complexity of disjointness.
\newblock {\em Theor. Comput. Sci.}, 106(2):385--390, 1992.

\bibitem[Raz95]{Raz}
Ran Raz.
\newblock A parallel repetition theorem.
\newblock In {\em Proceedings of the Twenty-Seventh Annual ACM Symposium on Theory of Computing}, STOC '95, page 447–456, New York, NY, USA, 1995. Association for Computing Machinery.

\bibitem[RR15]{RR}
Sivaramakrishnan~Natarajan Ramamoorthy and Anup Rao.
\newblock How to compress asymmetric communication.
\newblock In {\em Proceedings of the 30th Conference on Computational Complexity}, CCC '15, page 102–123, Dagstuhl, DEU, 2015. Schloss Dagstuhl--Leibniz-Zentrum fuer Informatik.

\bibitem[RY20]{RaoYehudayoff-text}
Anup Rao and Amir Yehudayoff.
\newblock {\em Communication Complexity: and Applications}.
\newblock Cambridge University Press, 2020.

\bibitem[She18]{Sherstov}
Alexander~A. Sherstov.
\newblock Compressing interactive communication under product distributions.
\newblock {\em SIAM Journal on Computing}, 47(2):367--419, 2018.

\bibitem[SK87]{KalyanasundaramSchnitger87}
Georg Schnitger and Bala Kalyanasundaram.
\newblock The probabilistic communication complexity of set intersection.
\newblock In {\em Proceedings of the Second Annual Conference on Structure in Complexity Theory, Cornell University, Ithaca, New York, USA, June 16-19, 1987}, pages 41--47. {IEEE} Computer Society, 1987.

\bibitem[ST23]{sudakov2023matrix}
Benny Sudakov and István Tomon.
\newblock Matrix discrepancy and the log-rank conjecture.
\newblock \url{https://arxiv.org/abs/2311.18524}, 2023.

\bibitem[Yu22]{huacheng}
Huacheng Yu.
\newblock Strong {XOR} lemma for communication with bounded rounds : (extended abstract).
\newblock In {\em 63rd {IEEE} Annual Symposium on Foundations of Computer Science, {FOCS} 2022, Denver, CO, USA, October 31 - November 3, 2022}, pages 1186--1192. {IEEE}, 2022.

\end{thebibliography}

\end{document}